\documentclass{article} 
\usepackage{amsmath}
\usepackage{graphicx}
\usepackage{amssymb}
\usepackage{hyperref}  
\usepackage{cleveref}
\usepackage{ntheorem}

\usepackage{multirow}
\usepackage{array}
\usepackage{color}
\usepackage{epsfig} 
\allowdisplaybreaks

\newtheorem{theorem}{Theorem}
\newtheorem*{proof}{Proof}

\newtheorem{lemma}[theorem]{Lemma}

\newtheorem{definition}[theorem]{Definition}
\newtheorem{proposition}[theorem]{Proposition}
\newtheorem{example}[theorem]{Example}
\newtheorem{remark}[theorem]{Remark}

\newcommand{\eps}{\varepsilon}

\newcommand{\supro}{\mbox{$\sup {\rm RO}$}}

\newcommand{\lfteqn}{\vspace{-0.08in} \begin{eqnarray} \begin{array}{lllllll}}
\newcommand{\ndeqn}{\vspace{-0.08in} \end{array} \nonumber \end{eqnarray}}
\newcommand{\Lfteqn}{\vspace{-0.08in} \begin{eqnarray} \begin{array}{lllllll}}
\newcommand{\Ndeqn}{\vspace{-0.08in} \end{array}  \end{eqnarray}}

\title{\bf Supervisory Control of Multi-Agent Discrete-Event Systems\\
 with Partial Observation}

\author{Yingying Liu, Jan Komenda, and Zhiwu Li 
\thanks{
The work of Jan Komenda is supported by RVO 67985840 and GA{\v C}R
grant 19-06175J.}
\thanks{Yingying Liu is with Shanghai Jiao Tong University, Shanghai, 200240, China, Zhiwu Li is with School of Electro-Mechanical Engineering, Xidian University, Xi'an, 710071, China,
and Jan Komenda is with Institute of Mathematics, Academy of Sciences of the Czech
Republic, Prague, Czech Republic.}}

\begin{document}
\maketitle \thispagestyle{empty} \pagestyle{empty}

\begin{abstract}

In this paper we investigate multi-agent discrete-event systems with partial observation. The agents can be divided into several groups  in each of which the agents have identical (isomorphic) state transition structures, and thus can be {\it relabeled} into the same template. Based on the template a {\it scalable supervisor} whose state size and computational cost are independent of the number of agents is designed for the case of partial observation. The scalable supervisor under partial observation does not need  to be recomputed regardless of how many agents are added to or removed from the system. 
We generalize our earlier results to partial observation  by proposing sufficient conditions for safety and maximal permissiveness of  the scalable least restrictive supervisor on the template level.    An example is provided to illustrate the proposed scalable supervisory synthesis.
\end{abstract}


\section{Introduction}\label{intro}
In manufacturing, logistical or similar technological systems one often encounters the need 
of not specifying the number of subsystems a priori, since their number can vary  even without an upper bound. Very often, such components are just instantiations of a finite number of template subsystems, where each group of components is isomorphic to a template. 
Such systems, often called multi-agent discrete-event systems (DES), are used in situations where several entities (e.g. robots, machines) perform the same type of jobs, while their number can vary in time.
In multi-agent DES  the agents (modeled as subsystems) can be divided into several groups, and within each group the agents have similar or identical state transition structures.

The first work about multi-agent DES is found in \cite{RohloffLafortune2006}, where multi-agent DES with a single group of isomorphic agents  (i.e. single template) are considered.
A more general multi-agent DES framework has been presented in \cite{Su13}, where a broadcasting-based parallel composition rule is used to describe  both cooperative and competitive interactions of agents.  In \cite{Jiao17} the authors have studied modular discrete-event systems that are formed as synchronous products of components that allow several isomorphic agents. The authors exploit the symmetry using state tree structures for achieving a larger computational benefit.
In \cite{SuLennartsson17} all local requirements are also instantiated from a given requirement template. A control protocol synthesis is investigated, which assumes that each private alphabet is only observable to the corresponding agent, but the global alphabet is accessible by all agents.

In \cite{Automatica} the case of several templates is studied under the assumption that there are no shared events inside the groups (among isomorphic agents), but also no shared events among different groups (represented by templates). Moreover, the method in \cite{Automatica} only deals with the complete observations case.
Our goal is to synthesize supervisors for templates that guarantee a specification regardless of the  number of  agents corresponding to templates (i.e. number of agents in different groups). 

This work extends the results of \cite{Automatica} in several direction. Firstly, we generalize the computation of scalable supervisors to the case of partial observation, where observability of the specification language is needed to synthesize the supervisors. Secondly, we compare   permissiveness of a monolithic supervisor with a scalable, template based, supervisor under partial observation.  We investigate under which conditions a scalable least restrictive supervisor  based on supremal relatively observable sublanguage on the template level is not more restrictive than the monolithic supervisor. Finally, we relax the assumption about absence of shared events and allow shared events among different groups, i.e. the templates are allowed to have shared events.

The paper is organized as follows. The next
section  contains  preliminaries about supervisory control and recalls basic notions and results used throughout this paper.  Section 3  is devoted to investigation of sufficient conditions for existence of scalable safety supervisor based on the supervisor for the relabeled system. In Section 4 we 
propose sufficient conditions for the scalable least restrictive supervisor to be as permissive as the monolithic supervisor.
 Section 5 reaches  conclusions.

\section{Preliminaries}\label{pre}

Let the DES plant to be controlled be modeled by a {\it generator}
${\bf G}=(Z,\Sigma,\delta,z_0)$,  
where $\Sigma$ is a finite event set, $Z$ is the finite state set, $z_0\in Z$ the initial state,  and $\delta:Z\times \Sigma\rightarrow Z$ the (partial) transition function. Extend $\delta$ in the usual way such that $\delta:Z\times \Sigma^{*}\rightarrow Z$. The {\it closed behavior} of ${\bf G}$ is the language
$L({\bf G}):=\{s\in \Sigma^{*}\mid \delta(z_0,s)!\}$, where the notation $\delta(z_0,s)!$ means that $\delta(z_0,s)$ is defined. 
 A string $s_1\in \Sigma^{*}$ is a {\it prefix} of another string $\textit{s}\in \Sigma^{*}$, written $s_1\leq s$, if there exists $s_2\in \Sigma^{*}$ such that $s_1s_2$ = $s$. The length of string $s$ is denoted by $|s|$.  
 We say that $L({\bf G})$ is {\em prefix-closed\/} if every for string $s\in L({\bf G})$ every prefix $s_1\leq s$ is also in
 $L({\bf G})$.

For partial observation, let the event set $\Sigma$ be partitioned into $\Sigma_o$, the observable event subset, and $\Sigma_{uo}$, the unobservable subset (i.e. $\Sigma=\Sigma_o \dot{\cup} \Sigma_{uo}$). A {\em (natural) projection} $P: \Sigma^* \to \Sigma_o^*$ is defined according to
\begin{flushleft}
\centering{\ \ \ \ \ \ \ \ \ \ \ \ $P(\eps)=\eps,\ \eps \mbox{ is the empty string;} $}
$$ P(\sigma)=\left\{
\begin{aligned}
\eps, \mbox{ if}\ \sigma\notin \Sigma_o, \\
\sigma, \mbox{ if}\ \sigma\in \Sigma_o;
\end{aligned}
\right.
$$ 
\centering{\ \ \ \ \ \ \ \ \ \ \ \ \ \ \ $P(s\sigma)=P(s)P(\sigma), s\in \Sigma^*,\sigma\in \Sigma.$}
\end{flushleft}
In the usual way, $P$ is extended to $P :  Pwr(\Sigma^*) \to Pwr(\Sigma_o^*)$, where $Pwr(\cdot)$ denotes powerset. The {\em inverse image} of $P$, denoted by $P^{-1} : 
Pwr(\Sigma_o^*) \to Pwr(\Sigma^*)$, is defined as $P^{-1}(s)=\{w\in \Sigma^* 
\mid P(w) = s\}$. The definitions can naturally be extended to 
languages. The projection of a generator ${\bf G}$ is a generator $P({\bf G})$ whose 
behavior satisfies $L(P({\bf G}))=P(L({\bf G}))$ and $L(P({\bf G}))=P(L({\bf G}))$.
More details about partially observed DES can be found in \cite{CL08}. 

Fixing a reference sublanguage $C  \subseteq L({\bf G})$, we introduce relative observability of language $K$.
Let $K \subseteq C \subseteq L$. $K$ is {\em $C$-observable\/} with respect to  $L$ and $\Sigma_o$ if 
\begin{align*} 
&(\forall s,s'\in \Sigma^*) (\forall \sigma \in \Sigma)  (s\sigma \in K\ \&\ s' \in C\ \&\ s'\sigma \in L({\bf G}) \\
& \&\ P(s)=P(s'))\Rightarrow  s'\sigma \in K.
\end{align*}
  The following Lemma, transitivity of relative observability, is needed in the proof of Theorem~\ref{Theorem:main1}).
\begin{lemma} \label{transitivityRO}
Suppose $K\subseteq N\subseteq L=L(G)\subseteq \Sigma^*$ and reference languages $C\subseteq C'$. If $K$ is relatively observable with respect to $C$ and $N$ and $N$ is relatively observable with respect to  $C'$ and $L$, then $K$ is relatively observable with respect to $C$ and $L$.
 \end{lemma}
\begin{proof} Let $sb\in K$, $s'\in C$, $P(s)=P(s')$, $b\in \Sigma$, and $s'b\in L$. We need to show that $s'b\in K$. Since $K\subseteq N$, i.e. $K\subseteq N$ as well, we have $sb\in N$.  Similarly, $C\subseteq C'$ implies $C\subseteq C'$. Thus, we have $s'\in C'$. Since $P(s)=P(s')$ and $s'b\in L$ it follows from $N$ being relatively observable with respect to  $C'$ and $L$ that $s'b\in N$. Finally, we obtain from relative observability of $K$ with respect to $C$ and $N$ that $s'b\in K$, which shows that  is relatively observable with respect to $C$ and $L$. 
\end{proof}

It is proved in \cite{Cai13} that relative observability is closed under arbitrary set unions. Given a generator ${\bf G}$ with $\Sigma=\Sigma_o\dot{\cup}\Sigma_{uo}$. Let $E\subseteq \Sigma^*$ be a specification language for ${\bf G}$. The family of all sublanguages of $E$ that are $E$-observable  with respect to $L$ and $\Sigma_o$ is
$\mathcal{O}(E,L) := \{K\subseteq E | K  \mbox{ is} \ E\mbox{-observable wrt }L \mbox{and}\ \Sigma_o\}.$
Then $\mathcal{O}(E,L)$ has a unique supremal element \cite{Cai13}, i.e.
\begin{align*}
\sup\mathcal{O}(E,L)=\cup \{K|K\in \mathcal{O}(E,L)\}.
\end{align*}

To describe the structure of a multi-agent plant ${\bf G}$, we briefly review a concept, called {\it relabeling map}. We refer the reader to \cite{Automatica} and reference therein for a more complete treatment for the relabeling map.
Let $T$ be a set of new events, i.e. $\Sigma \cap T=\emptyset$, and $R: \Sigma\rightarrow T$. Define a \textit{relabeling} map $R: \Sigma\rightarrow T$ such that is surjective  but need not be injective.
We require that events $\sigma_1, \sigma_2\in \Sigma$ with the same $R$-image, i.e. $R(\sigma_1)= R(\sigma_2)=\tau\in T$, have similar physical meaning and are just two instantiations of $\tau$ in different but isomorphic subsystems.
We extend $R$ by morphism to $R: \Sigma^{*}\rightarrow T^{*}$ according to

(i) $R(\varepsilon) = \varepsilon$, $\varepsilon$ is the empty string;

(ii) $R(s\sigma) = R(s)R(\sigma)$, $\sigma\in\Sigma$, $s\in \Sigma^*$.

\noindent Note that $R(s) \neq \varepsilon$ for all $s \in \Sigma^* \setminus \{\varepsilon\}$.
Further extend $R$ for languages, i.e. $R:Pwr(\Sigma^{*})\rightarrow Pwr(T^{*})$, and define
$R(L)=\{R(s) \in T^* | s\in L\},\ \ L\subseteq \Sigma^{*}.$ 
The \textit{inverse-image function} $R^{-1}$ of $R$ is given by $R^{-1}:$ $Pwr(T^{*})\rightarrow Pwr(\Sigma^{*})$:
 $R^{-1}(H)=\{s\in \Sigma^{*}| R(s)\in H\}$, \ $H\subseteq T^{*}$.
Note that $RR^{-1}(H)=H$, $H\subseteq T^{*}$;
while $R^{-1}R(L)\supseteq L$, $L\subseteq \Sigma^{*}$.
We say that $L \subseteq \Sigma^*$ is  (${\bf G}, R$)-{\it normal} if $R^{-1}R(L) \cap L({\bf G}) \subseteq L$. 

We now discuss computation of $R$ and $R^{-1}$ by generators.
Let \\
${\bf G}=(Q,\Sigma,\delta,q_0)$ be a generator.  First, relabel each transition of ${\bf G}$ to obtain ${\bf G}_T = (Q,T,\delta_T,q_0)$, where $\delta_T : Q \times T \rightarrow Q$ is defined by
$\delta_T(q_1, \tau) = q_2 \mbox{ iff } (\exists \sigma \in \Sigma) R(\sigma)=\tau \ \&\ \delta(q_1,\sigma)=q_2.$
Hence $L({\bf G}_T)=R(L({\bf G}))$. However, ${\bf G}_T$ as given above may be {\it nondeterministic} \cite{Wonham16}. Thus apply {\it subset construction} \cite{Wonham16} to convert ${\bf G}_T$ into a deterministic generator ${\bf H}=(Z,T,\zeta,z_0)$, with $L({\bf H})=L({\bf G}_T)$. 

\begin{figure}
  \centering
  \includegraphics[width=0.18\textwidth]{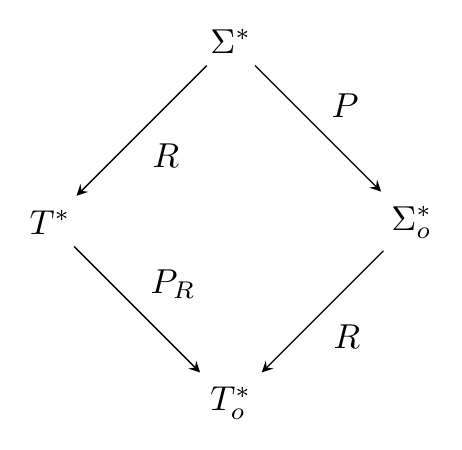}
  \caption{Schematic of natural projection $P$ and relabeling map $R$}
\label{fig:nr}
\end{figure}
We require that $R$ preserves observability status of events in $\Sigma$. Thus,  
 $T_o:=\{R(\sigma)|\sigma \in \Sigma_o\}$, $T_{uo}:=\{R(\sigma)|\sigma \in \Sigma_{uo}\}$, and $T=T_o\dot{\cup}T_{uo}$. 
Let $P: \Sigma^*\rightarrow\Sigma_o^*$ and $P_R: T^*\rightarrow T_o^*$ be natural projections, their relationships with $R$ are shown in~Fig.~\ref{fig:nr}.

Large complex DES are built out of the small ones using concurrent composition known as synchronous product.
Given languages $L_i\subseteq \Sigma_i^*, \; i=1,\dots,n$,
their synchronous product (parallel composition) is defined by 
$L_1\| \dots \|L_n=\bigcap_{i=1}^n P_i^{-1}(L_i)  \subseteq \Sigma^*$, where 
$P_i: \Sigma^*\to \Sigma_i^*$ are projections to local alphabets. There are corresponding definitions in terms of generators, i.e. $L({\bf G}_1 \| {\bf G}_2) = L({\bf G}_1) \| L({\bf G}_2)$\cite{CL08}.

\section{Scalable Supervisor under Partial Observation}\label{pro}

Let $R: \Sigma^{*}\rightarrow T^{*}$ be a relabeling map, and  the plant {\bf G} can be divided into $l\ (\geqslant1)$ groups of component agents. In each group $\mathcal{G}_i \,(i \in \{1,\ldots,l\})$ being a similar set of generators under the given relabeling map $R$, i.e. $\mathcal{G}_i = \{ {\bf G}_{i_1},\ldots,{\bf G}_{i \, n_i} \}$ ($n_i \geq 1$) and there is a template generator ${\bf H}_i$ such that
\begin{align} \label{eq:Hi}
(\forall j \in\{1,\dots,n_i\}) R({\bf G}_{i_j}) = {\bf H}_i.
\end{align}
Let ${\bf G}_{i_j}$ be defined on $\Sigma_{ij}=\Sigma_{ij,o}\dot{\cup}\Sigma_{ij,uo}$,  templates ${\bf H}_i$ on $T_i=T_{i,o}\dot{\cup} T_{i,uo}$, and group languages $\mathcal{G}_i$ on $\Sigma_{i}=\bigcup_{j=1,...,n_i }\Sigma_{ij}$. Then  $R(\Sigma_{ij})=T_i$, $R(\Sigma_{ij,o})=R(\Sigma_{i,o})=T_{i,o}$, and $R(\Sigma_{ij,uo})=R(\Sigma_{i,uo})=T_{i,uo}$ for all $j \in \{1,...,n_i\}$. We emphasize that the number $n_i$ of generators (agents) in group $i$ is not fixed but may vary in time. We denote the alphabet of $\mathcal{G}_i$ by $\Sigma_i$, i.e.
$\Sigma_i=\cup_{j =1}^{n_i}\Sigma_{ij}$.
We require that $R$ preserves local status of
events in $\Sigma$; namely $R(\sigma)$ is an event in $T_i$ if and only if $\sigma\in \Sigma_i$. Thus $T_i:=\{R(\sigma)|\sigma \in \Sigma_i\}$ and $T=\bigcup_{i=1}^{i=l} T_i$. Similarly, $P_{i|R}: T^*\to T_i^*$, for $i=1,2$, are projections to local relabeled event sets.
The relationships of $R$, projections $P_i$, and $P_{i|R}$ are shown in Fig.~\ref{fig:nr2}.


Now we make the following assumptions.

\noindent (A1) The specification language $E$ is prefix closed and only generated languages are considered. 


\noindent (A2) The specification language $E$ is $({\bf G}, R)$-normal, i.e. $R^{-1}R(E)\cap L({\bf G}) \subseteq E$.

Due to (A1) all component agents are then automatically nonblocking.
Note that blocking issue can be handled in a similar way as in hierarchical control and has been solved for relabeling in \cite{IJC} by adapting an observer property (OP) in so called relabeled OP (ROP), but in this paper we focus on another extension of the framework in \cite{Automatica}, namely to partial observations
and maximal permissiveness.

\begin{figure}
  \centering
  \includegraphics[width=0.18\textwidth]{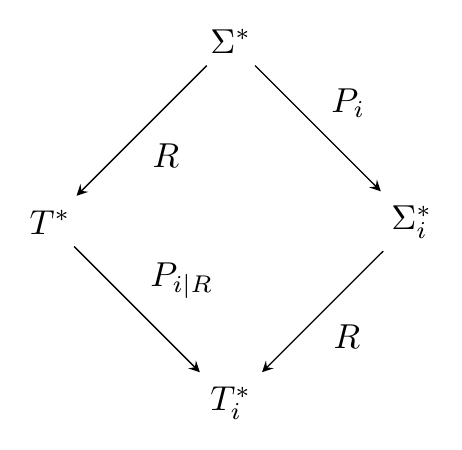}
  \caption{Schematic of  local projection  $P_i\; (i=1,\dots, l)$, $P_{i|R}$, and relabeling map $R$}
\label{fig:nr2}
\end{figure}

Given plant ${\bf G}$ and specification $E$, the monolithic supervisor under partial observation is computed based on supremal relatively observable sublanguages. The whole plant ${\bf G}$ is computed as synchronous product of all component agents:
\begin{align}\label{gi}
{\bf G} = ||_{i=1,\ldots,l} {\bf G}_i,\ \mbox{ where } {\bf G}_{i} = ||_{j=1,\ldots,n_i} {\bf G}_{i_j}.
\end{align}
Recall that ${\bf G} = ||_{i=1,\ldots,l} {\bf G}_i=\bigcap_{i=1,\ldots,l}P_i^{-1}( {\bf G}_i)$ for $P_i: \Sigma^*\rightarrow \Sigma_i^*$. Then supremal relatively observable sublanguage is computed:
\begin{align}
\label{SUPRO}
L({\bf SUP}_{o}) = \sup\mathcal{O}(E \cap L({\bf G}),L({\bf G})).
\end{align}

The supervisor ${\bf SUP}_{o}$  must be recomputed or reconfigured in order to adapt to the change of the number of agents (increases when more agents are added into the system to enhance productivity or decreases when some agents malfunction and are removed from the system). Therefore, in this paper we aim to synthesize scalable supervisors under partial observation whose state size is independent of the number of agents. 






\subsection{Scalable supervisory control with partial observation}\label{mainresult}
In this subsection we design a scalable supervisor ${\bf SSUP}_o$ that is independent of the number $n_i$ of agents for all $i \in \{1,\ldots,l\}$ and satisfies 
$\{\epsilon \} \subset$ $L({\bf SSUP}_o) \cap L({\bf G}) \subseteq L({\bf SUP}_o)$,
while the opposite inclusion (maximal permissiveness) is studied next.


Consider the plant {\bf G} as described in (\ref{gi}). Let $\Sigma =\Sigma_o \dot\cup \Sigma_{uo})$ be the event set of {\bf G}, and $R : \Sigma \rightarrow T$ a relabeling map. The procedure of designing a scalable supervisor under  partial observation is in steps (P1)-(P4), which  first synthesizes a supervisor for `relabeled system' under $R$ and then inverse-relabel the supervisor.
\smallskip

\noindent (P1) First compute the {\it relabeled plant}. Let $k_i\in \{1,2, \dots, n_i\}$ be the number of agents in group $i$ allowed to work in parallel, and compute the template ${\bf M}_i := R(||_{j=1,\dots,k_i} {\bf G}_{i_j})$. Then compute  the relabeled plant ${\bf M}$ as the synchronous product of the template generators ${\bf M}_i$, i.e.
\vspace{-3.0mm}
\begin{align} \label{eq:M}
{\bf M} := ||_{i=1,...,l} {\bf M}_i.
\end{align}

The event set of ${\bf M}$ is $T =T_o \dot\cup T_{uo}$, where $T_o = R(\Sigma_o)$, and $T_{uo}=R(\Sigma_{uo})$. 
Note that $k_i$ should be much smaller than $n_i$ (the number of agents in group $i$) for computational efficiency, it is trade-off between expressiveness and complexity . When $k_i=1$, we have ${\bf M}_i={\bf H}_i$ given by (\ref{eq:Hi}). Note that once $k_i$ are fixed, the
state sizes of ${\bf M}_i$ and ${\bf M}$ are fixed as well, and thus independent of the number $n_i$ of agents in group $i$. 
Recall that $k_i>1$ is desirable for allowing agents to work in parallel as 
$R({\bf G}_{i_1}\| {\bf G}_{i_2})$ can be strictly larger than 
$R({\bf G}_{i_1}) \|R({\bf G}_{i_2})=H_i \| H_i=H_i$,
which distinguishes relabeling (surjective mask) from natural projection.

\noindent (P2) Compute {\it relabeled specification} $F := R(E)$, where $E \subseteq \Sigma^*$ is the specification imposed on {\bf G}. 

\noindent (P3) Synthesize a {\it relabeled supervisor} under partial observation ${\bf RSUP}_{o}$ (a nonblocking generator) such that
\vspace{-3.0mm}
\begin{align*}
L({\bf RSUP}_{o}) = \sup\mathcal{O}(F,L({\bf M})) \subseteq T^*.
\end{align*}

\noindent (P4) Inverse-relabel ${\bf RSUP}_{o}$ to derive {\it scalable supervisor} ${\bf SSUP}_{o}$, i.e.
\vspace{-3.0mm}
\begin{align}\label{eq:SSUP}
{\bf SSUP}_{o} := R^{-1} ({\bf RSUP}_{o})
\end{align}

Notice that the computations involved in the above procedure are independent of the number $n_i$ ($i\in \{1,...,l\}$) of agents.
In (P1), once $k_i$ are fixed, the state sizes of ${\bf M}_i$ and ${\bf M}$ are fixed and independent of the number $n_i$ of agents in group $i$ (although dependent on $k_i$).
In (P3), the number of states of ${\bf RSUP}_{o}$ is independent of the number of agents as the state size of ${\bf M}$ is so.
Finally in (P4), inverse-relabeling does not change the number of states. Therefore ${\bf SSUP}_{o}$ has the same number of states as ${\bf RSUP}_{o}$. It then follows that  the state size of ${\bf SSUP}_{o}$
is independent of the number of agents in plant {\bf G}. 

Arguably, with large numbers of agents $k_i$ in the templates ${\bf M}_i$ and for a large number $l$ of modules (groups) it may be computationally challenging to compute  the global template ${\bf M}$. In such a situation we propose to combine the results of this paper
with coordination control approach, see e.g. \cite{KMvS15}, which consists in conditionally decomposing the template specification $F$, construct the corresponding coordinator and compute local supervisors for local templates ${\bf M}_i$ combined with the coordinator.
More formally, instead of computing $\sup\mathcal{CO}(F,L({\bf M})) $
we can first find a coordinator alphabet $\Sigma_k$ containing at least shared events
$\Sigma_s=\cup_{i,i'=1,\dots,l, \; i\not=i'} \Sigma_i\cap \Sigma_{i'}$ such that $F=\|_{i=1,\dots,l}P_{i+k}(F)$ is conditionally decomposable with respect to ``augmented'' local alphabets $\Sigma_i \cup \Sigma_k$, where $P_{i+k}:\Sigma^* \to (\Sigma_i \cup \Sigma_k)^*$ are the corresponding natural projection. Then the underlying coordinator is ${\bf M}_k=\|_{i=1,\dots,l} P_k({\bf M}_i)$, where
$P_k: \Sigma^* \to \Sigma_k^*$ is natural projection. Instead of computing supervisor for the whole template, i.e. $\sup\mathcal{CO}(F,L({\bf M})) $, we can compute supervisors
$\sup\mathcal{O}(P_{i+k}(F) ,L({\bf M}_i \| {\bf M}_k)) $. In general we only have $\|_{i=1,\dots,l}  \sup\mathcal{CO}(P_{i+k}(F) ,L({\bf M}_i \| {\bf M}_k)) \subseteq \sup\mathcal{CO}(F,L({\bf M}))$, but under some conditions (e.g. mutual observability between coordinated templates or conditions used in hierarchical control with partial observation) the equality holds.

 Now we compare the designed scalable supervisor for multi-agent systems with  the monolithic one. 
We emphasize that unlike modular or hierarchical control, where modular or
abstracted safety supervisor is always included in the monolithic one,
the situation here  is more complicated, because scalable supervisor is computed with respect to the template and not with respect to the relabeling of the plant $R(L({\bf G}))$, which would correspond to hierarchical control with relabeling being an abstraction. This is because the
relabeling map does not distribute with the synchronous product. In multi-agent systems we naturally use templates for control synthesis. Note that the inclusion 
\begin{equation*}
L({\bf M})=\|_{i=1}^l L({\bf M}_i)=\|_{i=1}^l R(||_{j=1,\dots,k_i} {\bf G}_{i_j}) \subseteq R(L({\bf G}))
\end{equation*}
 holds under the assumption that there are no shared events inside the group, and allow to share events among templates. The above inclusion
 is typically strict as the relabeling behaves differently than natural projection with respect to the
 synchronous product.
Hence, due to anti-monotonicity of supremal control operators in the plant argument, the scalable supervisor computed with respect to ${\bf M}$ as a plant can be larger than the supervisor (computed using $R({\bf G})$. Therefore the inclusion studied in the result below that establishes safety is also non trivial. 
 
\begin{theorem}\label{Theorem:main1} 
\label{supRO}
Suppose that (A1) and (A2) hold. If $ L({\bf M})$ is relatively observable with respect to $R(E \cap L({\bf G}))$ and $R( L({\bf G}))$, then $ L({\bf SSUP}_{o})\cap L({\bf G})\subseteq L({\bf SUP}_{o})$.
  \end{theorem}
\begin{proof}
Since $L({\bf SUP}_{o})=\supro (E\cap L({\bf G}),L({\bf G}))$, it suffices to prove that (i) $ L({\bf SSUP}_{o})\cap L({\bf G})\subseteq E \cap L({\bf G})$ and (ii) $L({\bf SSUP}_{o})\cap L({\bf G})$ is relatively observable with respect to $E \cap L({\bf G})$ and $L({\bf G})$. For (i) we have \\
$ L({\bf SSUP}_{o})\cap L({\bf G})=R^{-1}(L({\bf RSUP}_o))\cap L({\bf G})$\\ (by (P4))$=R^{-1}\supro (R(E),L({\bf M}))\cap L({\bf G})$
$\subseteq R^{-1}R(E) \cap L({\bf G})=E \cap L({\bf G})$ by assumption (A2). 

For (ii), let $sb\in L({\bf SSUP}_{o})\cap L({\bf G})$, $s'\in E \cap L({\bf G})\subseteq E\cap L({\bf G})$,  and $s'b\in L({\bf G})$. We need to show that $s'b\in L({\bf SSUP}_{o})\cap L({\bf G})$. We have 
$sb \in  L({\bf SSUP}_{o})\cap L({\bf G})
= R^{-1}\supro (R(E),L({\bf M}))\cap L({\bf G})$
thus $R(sb)\in \supro (R(E),L({\bf M}))\cap L({\bf G})$. We also have $R(s')\in R(E)$, and $R(s'b)\in R(L({\bf G}))$. Since $R(E \cap L({\bf G}))\subseteq R(E)$ and $\supro (R(E),L({\bf M}))$ is relatively observable with respect to $R(E)$ and $R( L({\bf G}))$, we get that $\supro (R(E),L({\bf M}))$ is relatively observable with respect to $R(E \cap L({\bf G}))$ and $L({\bf M})$\cite{Cai13}. By Lemma~\ref{transitivityRO},  combining it with relative observability of $L({\bf M})$ with respect to  $R(E \cap L({\bf G}))$ and $R(L({\bf G}))$, it implies that $\supro (R(E),L({\bf M}))$ is relatively observable with respect to $R(E \cap L({\bf G}))$ and $R(L({\bf G}))$. We  have 
$R(s'b)\in \supro (R(E),L({\bf M}))$. Therefore, $s'b\in R^{-1}R(s'b)  \subseteq R^{-1}\supro (R(E),L({\bf M}))\cap L({\bf G})=L({\bf SSUP}_{o})\cap L({\bf G})=L({\bf SSUP}_{o})\cap L({\bf G})$, which proves (ii). 
\end{proof}

Note that in this paper we  only discuss the  observability problem, since the controllability of the scalable supervisor has been well studied in \cite{Yingying18}.
Theorem~\ref{Theorem:main1}  provides a sufficient condition under which the scalable supervisor is always included in the monolithic one. This condition is $ L({\bf M})$ is  relatively observable with respect to $R(E \cap L({\bf G}))$ and $R( L({\bf G}))$. As shown above, this condition is essential in proving  relative observability of $L({\bf SSUP}_{o})\cap L({\bf G})$ with respect to $E \cap L({\bf G})$ and $L({\bf G})$. However, the computation of {\bf G} is required in this condition, which needs to be avoided. We will make an additional assumption and give the following result.

(SEF) We assume that the event sets of systems inside each group are pairwise disjoint, i.e. for all $i \in \{1,\ldots,l\}$ and $j,j' \in \{1,\ldots,k_l\}$, 
$\Sigma_{ij}\cap \Sigma_{ij'}=\emptyset$. 

Note that unlike \cite{Automatica} we allow shared events between templates.
\begin{proposition} \label{prop:check}
Let (SEF) hold.
If for each group $i \in \{1,\ldots,l\}$ and ${\bf G}_{i_1},{\bf G}_{i_2}\in \mathcal{G}_i$, $L({\bf H}_{i})$ is  relatively observable with respect to $R(P_i(E)\cap L({\bf G}_{i_1}\|{\bf G}_{i_2}))$ and $R(L({\bf G}_{i_1}\|{\bf G}_{i_2}))$, then $L({\bf M})$ is  relatively observable with respect to $R(E\cap L({\bf G})$ and $R(L({\bf G}))$.
  \end{proposition}
The proof of Proposition~\ref{prop:check} is referred to Section~5. Proposition~\ref{prop:check} indicates that relative observability of  $L({\bf M})$ with respect to $R(E\cap L({\bf G}))$ and $R(L({\bf G}))$ can be verified by checking  relative observability of $L({\bf H}_{i})$ for each group with respect to only two (arbitrarily chosen) component agents.
 Therefore, the computational effort of checking the sufficient condition in Theorem~\ref{Theorem:main1} is low.  We recall that algorithms for checking relative observability and computing supremal relatively observable sublanguages are proposed in  \cite{Cai2018} and \cite{ACB17}.   
An illustrative example is given below.
\begin{example}
Consider a small factory consisting of $n_1$ input machines \\
${\bf G}_{1_1}, \dots, {\bf G}_{1_{n_1}}$ and $n_2$ output machines ${\bf G}_{2_1},  \dots, {\bf G}_{2_{n_2}}$, linked by a buffer in the middle. The generators of the agents are shown in Fig.~\ref{fig:SmallFact22}. Based on their different roles, the machines are divided into 2 groups:
$\mathcal{G}_1=\{{\bf G}_{1_1},\dots, {\bf G}_{1_{n_1}}\}$ and $\mathcal{G}_2=\{{\bf G}_{2_1},\dots,{\bf G}_{2_{n_2}}\}$.
Let the relabeling map $R$ be given by
\begin{align*}
& R(1i0)=10,\ R(1i1)=11,\ R(1i2)=12,\ R(1i3)=13,\\
& R(2j0)=20,\ R(2j1)=21,\ R(2j2)=22,\ R(2j3)=23\ 
\end{align*}
with template generators ${\bf H}_1$ and ${\bf H}_2$ respectively, where \\
$\Sigma = \Sigma_o \dot\cup \Sigma_{uo} = \{1i0,1i1,1i3,2j0,2j1,2j3\} \dot\cup \{1i2,2j2 \}$, $i \in [1,n_1]$, and  $j\in [1,n_2]$. 

Since the event sets of agents in the same group are pairwise disjoint 
($\Sigma_1\cap \Sigma_2=\{1i0,1i1,1i2,1i3\}\cap \{2i0,2i1,2i2,2i3\}=\emptyset$), Assumptions (A1) and (SEF) hold for this example. The specification is to  avoid underflow and overflow of buffer with two capacities, which is enforced by a generator {\bf E} and shown in Fig.~\ref{fig:spec} (lower part). It can be verified that ${\bf E}=R^{-1}R({\bf E})$, we thus have Assumption  (A2) holds. Now we design templates for each group. Let $k_1=2, k_2=1$. The templates are $L({\bf M}_1) := L(R({\bf G}_{11}\|{\bf G}_{12}))$ (for the input group) and $L({\bf M}_2) := L(R({\bf G}_{2_1}))$ (for the output group).  Note that in this example we have
$L({\bf M_1})=R(L({\bf G}_{1_1}\| {\bf G}_{1_2})))$ is different from $L({\bf H_1})=R(L({\bf G}_{1_1}))$, while $L({\bf M_2})=L({\bf H_2}).$

Now we need to check the sufficient condition of Theorem~\ref{Theorem:main1}; namely, $L({\bf M})$ is relatively observable with respect to $R(E\cap L({\bf G}))$ and  $R(L({\bf G}))$. If  this condition holds, then we can employ procedure (P1)-(P4) proposed in Section~3.1 to compute the scalable supervisor. By Proposition~\ref{prop:check}, we need to check if for $i=1,2$, $L({\bf H}_i)$ is relatively observable with respect to $R(P_i(E)\cap L({\bf G}_{i_1}\| {\bf G}_{i_2}))$ and  $R(L({\bf G}_{i_1}\| {\bf G}_{i_2}))$. 
We compute $sup\mathcal{CO}(R(P_i(E)\cap L({\bf G}_{i_1}\| {\bf G}_{i_2}))$, $L({\bf H}_i)$ and get  that 
$L({\bf H}_i)$ is relatively observable with respect to $R(P_i(E)\cap L({\bf G}_{i_1}\| {\bf G}_{i_2}))$ and  $R(L({\bf G}_{i_1}\| {\bf G}_{i_2}))$.  We also get that $R(P_i(E))\cap R(L({\bf G}_{i_1}\| {\bf G}_{i_2}))=R(P_i(E))\cap R(L({\bf G}_{i_1}\| {\bf G}_{i_2}))$. Hence by Proposition~\ref{prop:check}, we have that $L({\bf M})$ is relatively observable with respect to $R(E\cap L({\bf G}))$ and  $R(L({\bf G}))$. Therefore the sufficient condition of Theorem~\ref{Theorem:main1} is satisfied.

In procedure (P1)-(P4), we design a scalable supervisor ${\bf SSUP}_o$, displayed in Fig.~\ref{fig:supro}. It can be verified that the scalable supervisor is included in the supremal monolithic one, i.e. $ L({\bf SSUP}_o)\cap L({\bf G})\subseteq L({\bf SUP}_o)$.
 
\begin{figure}
  \centering
  \includegraphics[width=0.7\textwidth]{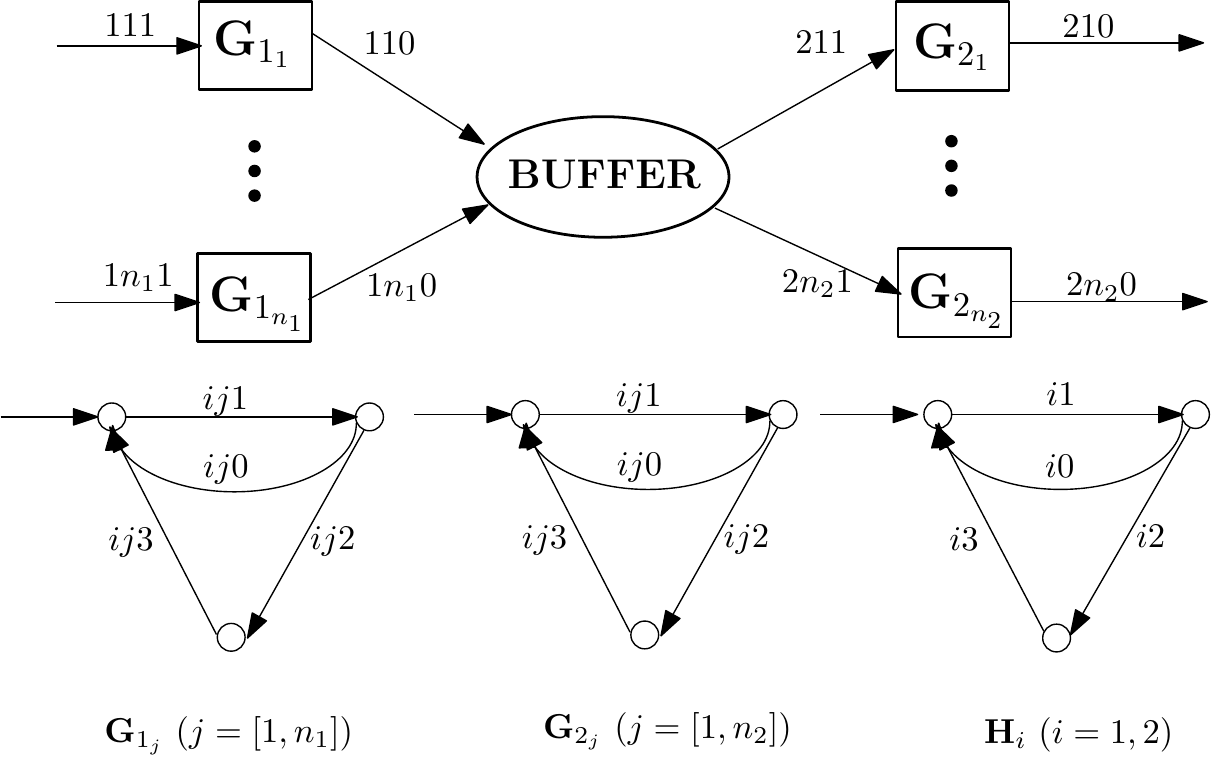}
  \caption{Small factory: system configuration and component agents.  Events $1j1$ ($ j \in [1,\dots,n_1]$)  and $1j2$ ($ j \in [1,\dots,n_2]$) mean that machine ${\bf G}_{i_j}$ starts to work by taking in a workpiece; events $1j0$ and $2j0$ mean that ${\bf G}_{i_j}$ finishes work and outputs a workpiece; events $1j2$ and $2j2$ mean that ${\bf G}_{i_j}$ is broken down; events $1j3$ and $2j3$ mean that ${\bf G}_{i_j}$ is repaired. Convention: the initial state of a generator is labeled by a circle with an entering arrow.
The same notation will be used in subsequent figures.}
\label{fig:SmallFact22}
\end{figure}
\begin{figure}
  \centering
  \includegraphics[width=0.45\textwidth]{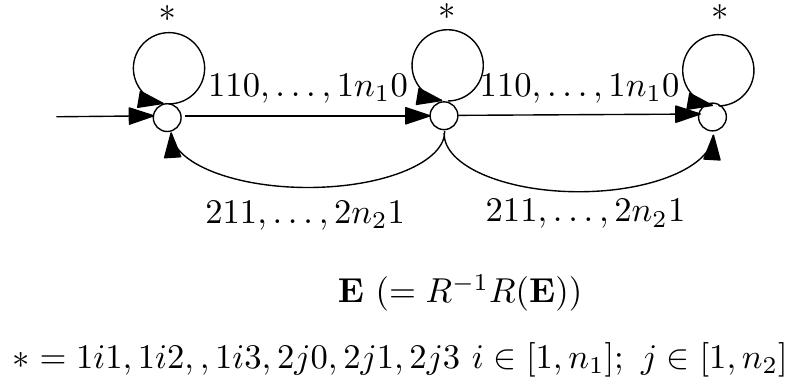}
  \caption{Small factory: templates and specification generator}
\label{fig:spec}
\end{figure}
\end{example} 
\begin{figure}
  \centering
  \includegraphics[width=0.7\textwidth]{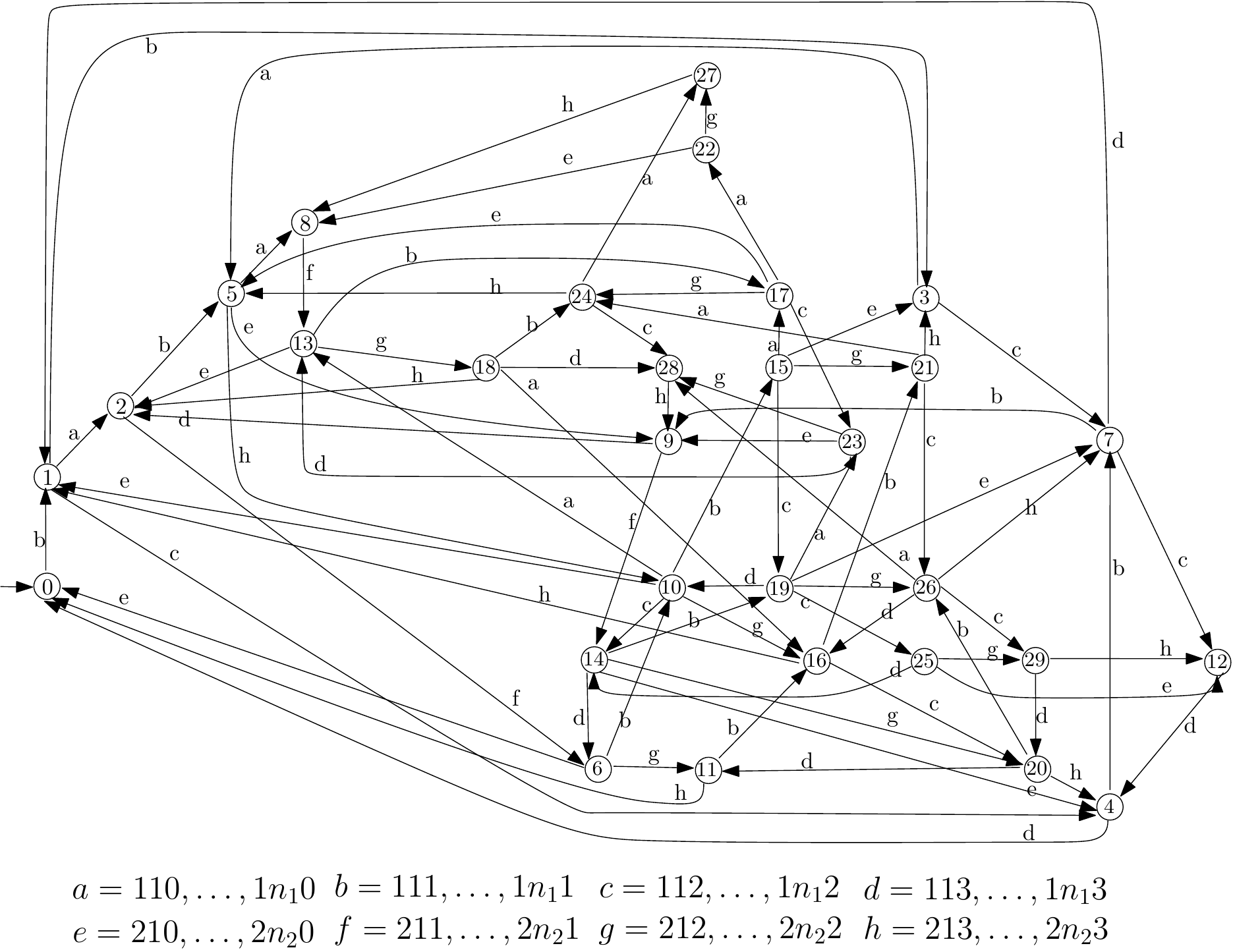}
  \caption{Small factory:  ${\bf SSUP}_o$}
\label{fig:supro}
\end{figure}

\section{Maximal permissiveness of the scalable supervisor}
In the last section, we provid the conditions that the obtained relatively observable sublanguage is a sublanguage of the supremal relatively observable sublanguage, i.e. $ L({\bf SSUP}_o)\cap L({\bf G})\subseteq L({\bf SUP}_o)$. In this suction, we will propose conditions for the opposite inclusion, i.e. the maximal permissiveness of the scalable supervisor.
To give our main result, the following conditions are  needed.
    \smallskip
    \begin{definition}[ Relabeling observational consistency]
        A relabeling map $R$  is said to be {\em Relabeling observation consistent} (ROC) with respect to plant language $L=L({\bf G})\subseteq \Sigma^*$ and natural projection $P$  if for all strings $s\in L$ and $t'\in R(L)$ such that $P_{R}(R(s))=P_{R}(t')$, there exists a
string $s'\in L$ such that $R(s')=t'$ and $P(s)=P(s')$.
  \end{definition}
  \begin{figure}
  \centering
  \includegraphics[width=0.55\textwidth]{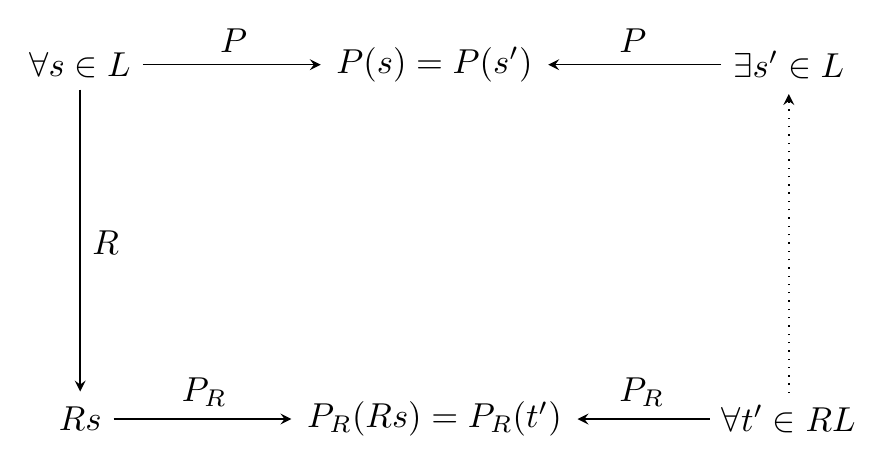}
  \caption{Illustration of ROC condition}
\label{ROC}
\end{figure}
  As we can see from figure \ref{ROC}, ROC condition of a language requires that for every string $s$ in the language and every string $t'$ in the relabeling of this language that looks the same as the relabeling of $s$,
  there exists another string $s'$ in the language that looks the same as $s$ and its relabeling equals $t'$. 
We will see in section 5 that ROC always holds if there are no shared events between different agents in the same group.   
    \smallskip
          \begin{definition}[Local relabeling observer consistency]
          A relabeling map $R$  is said to be {\em Locally relabeling observation consistent} (LROC) with respect to plant language $L=L({\bf G})\subseteq \Sigma^*$ and natural projection $P$  if 
        for all strings $s,s'\in L$ with $P(s)=P(s')$ and unobservable events $b,b'\in \Sigma_{uo}$ such that $R(b)=R(b')$  we have
$$\left( sb\in L \wedge  s'b' \in L \right)\; \Rightarrow \; s'b \in L .$$
  \end{definition}
  
    %
The following result formulates sufficient conditions for maximal permissiveness of the scalable supervisor.
  \begin{theorem}
 \label{maximal permissiveness}
Let $L=L({\bf G})$ be ROC and LROC with respect to relabeling map $R$ and natural projection $P$
and let specification $E\subseteq \Sigma^*$ be  prefix-closed and $({\bf G}, R)$-normal and let $L({\bf M}) \subseteq  RL({\bf G})$.
Then $ L({\bf SUP}_{o}) \subseteq L({\bf SSUP}_{o})\cap L({\bf G})$
\end{theorem}
\begin{proof}
We need to show that
$$\supro (E\cap L({\bf G}),L({\bf G}))  \subseteq  R^{-1}\supro (R(E),L({\bf M})) \cap L({\bf G}) .$$
Since by definition  $\supro (E\cap L({\bf G}),L({\bf G}))  \subseteq E\cap L({\bf G}) $,
it suffices to show that 
$R(\supro (E\cap L({\bf G}),L({\bf G})) ) \subseteq  \supro (R(E),L({\bf M}))$.
Note that due to the assumption $L({\bf M}) \subseteq  RL({\bf G})$ every language $N\subseteq T^*$ that is relatively observable wrt $R(E)$ and $RL({\bf G})$  is also relatively observable wrt $R(E)$ and $L({\bf M}))$.
Therefore,
$\supro (R(E),RL({\bf G}))  \subseteq \supro (R(E),L({\bf M}))$, a standard antimonotonicity property of supremal sublanguages of supervisory control in the plant argument.
Thus, we only need to show
$R(\supro (E\cap L({\bf G}),L({\bf G}))  )\subseteq \supro (R(E),RL({\bf G})) $.
The latter inclusion then amounts to show that \\
(i) $R(\supro (E\cap L({\bf G}),L({\bf G})))  \subseteq R(E)$, which is obvious, and\\
(ii) $R(\supro (E\cap L({\bf G}),L({\bf G})))$ is relatively observable wrt $R(E)$ and $RL({\bf G})$.
To show (ii), 
let $t,t'\in T^*$ be such that $P_R(t)=P_R(t')$, and let $a\in T$ be such that $ta \in R\supro (E\cap L({\bf G}),L({\bf G}))$, $t'\in R(E)$, and $t'a \in RL({\bf G})$. We have to show that $t'a\in R(\supro (E\cap L({\bf G}),L({\bf G})))$. 

From $ta\in R(\supro (E\cap L({\bf G}),L({\bf G})))$ we have that there exists \\
$sb\in \supro (E\cap L({\bf G}),L({\bf G}))$ such that $R(sb)=ta$. Since $t'a \in RL({\bf G})$ and $P_R(R(sb)) = P_R(t'a)$, ROC (applied to $sb$ playing the role of $s$ and $t'a$ playing the role of $t'$) implies that 
there is a $v' \in L({\bf G})$ such that $R(v') = t'a$ and $P(v')=P(sb)$. Then $v'=s'b'$ for some $s'\in L({\bf G})$ and $b'\in \Sigma$ such that $R(b')=a$ and $P(b')=P(b)$. 
This means that if $b'\in \Sigma_o$ then $b'=b$.
Note that $R(s')=t'\in  R(E)$, i.e. by normality assumption (A2) we have 
$s'\in  R^{-1}R(E) \cap L({\bf G}) =E$.
Hence, in case $b=b'\in \Sigma_o$ w have that
$sb\in \supro (E\cap L({\bf G}),L({\bf G}))$, $s'\in E\cap L({\bf G})$, $s'b\in L({\bf G}))$, and $P(s')=P(s)$.
Thus, from relative observability of $sb\in \supro (E\cap L({\bf G}),L({\bf G}))$ with respect to
$E\cap L({\bf G})$ and $L({\bf G})$ we obtain 
$s'b\in \supro (E\cap L({\bf G}),L({\bf G}))$, whence $t'a=R(s'b)\in R(\supro (E\cap L({\bf G}),L({\bf G})))$.

Now, if both $b,b'\in \Sigma_{uo}$ are unobservable, then we 
 have  
$s,s'\in L({\bf G}) \wedge P(s')=P(s) \wedge sb \in L({\bf G}) \wedge s'b'\in L({\bf G}) \wedge R(b)=R(b') $.
From LROC we get  $s'b \in L({\bf G})$ as well. Thus,  we can continue in the same way as with
 $b,b'$  both observable.

\end{proof}

Note that the condition $L({\bf M}) \subseteq  RL({\bf G})$ in Theorem~\ref{maximal permissiveness} always holds under the assumption (SEF), which is shown by the following lemmas.  
 \begin{lemma}
 \label{rp-relation}
Let (SEF) hold. For any string $s\in \Sigma_{i}^*$, natural projections $P_i:\Sigma^*\rightarrow \Sigma_i^*$, $P_{R,i}:T^*\rightarrow T_i^*$, and the relabeling map $R:\Sigma^*\rightarrow T^*$, we have $RP_i^{-1}(s)=P_{R,i}^{-1}R(s)$.
\end{lemma}
\begin{proof}
By induction, intuitively $RP_i^{-1}(\epsilon)=P_{R,i}^{-1}R(\epsilon)=\epsilon$. Suppose that $RP_i^{-1}(s)=P_{R,i}^{-1}R(s)$ holds. Then for any $\sigma\in \Sigma_i$, we need to show $RP_i^{-1}(s\sigma)=P_{R,i}^{-1}R(s\sigma)$. Since  $RP_i^{-1}(s\sigma)=R(P_i^{-1}(s)P_i^{-1}(\sigma))=RP_i^{-1}(s)RP_i^{-1}(\sigma)$ and $P_{R,i}^{-1}R(s\sigma)=P_{R,i}^{-1}(R(s)R(\sigma))=P_{R,i}^{-1}R(s)P_{R,i}^{-1}R(\sigma)$. By recalling $RP_i^{-1}(s)=P_{R,i}^{-1}R(s)$, we only need to show $RP_i^{-1}(\sigma)=P_{R,i}^{-1}R(\sigma)$. Since $P_i:\Sigma^*\rightarrow \Sigma_i^*$ and $\sigma\in \Sigma_i$, we have that $RP_i^{-1}(\sigma)=R((\Sigma\setminus \Sigma_i)^*\sigma (\Sigma\setminus \Sigma_i)^*)=(T\setminus T_i)^*R(\sigma)(T\setminus T_i)^*$. Similarly, $P_{R,i}:T^*\rightarrow  T_i^*$, so we get $P_{R,i}^{-1}R(\sigma)=(T\setminus T_i)^*R(\sigma)(T\setminus T_i)^*$. Thus,  $RP_i^{-1}(s)=P_{R,i}^{-1}R(s)$.
\end{proof}

 \begin{lemma}
 \label{mg-inclution}
$L({\bf M}) \subseteq  RL({\bf G})$ always holds under the assumption (SEF).
\end{lemma}
\begin{proof}
For any string $t\in L({\bf M})\subseteq T^*$ we need to show $t\in RL({\bf G})$. By $t\in L({\bf M})=\|_{i=1}^{l} L({\bf M}_i)=\bigcap_{i=1}^{l} P_{R,i}^{-1}(L({\bf M}_i))$, we have $P_{R,i}(t)\in L({\bf M}_i)=R(\|_{j=1}^{k_i} L({\bf G}_{ij}) )$. Then there exist strings $s_i\in \Sigma_i^*$ such that  $s_i\in \|_{j=1}^{k_i} L({\bf G}_{ij}) $ and $P_{R,i}(t)=R(s_i)$. Then we get that $\|_{i=1}^{l} s_i\in \|_{i=1}^{l} \|_{j=1}^{k_i} L({\bf G}_{ij}) $. Since 
$\|_{i=1}^{l} s_i=\bigcap_{i=1}^{l} P_{i}^{-1}(s_i)$, we obtain $\bigcap_{i=1}^{l} P_{i}^{-1}(s_i)\in \|_{i=1}^{l} \|_{j=1}^{k_i}L({\bf G}_{ij}) $. We apply $R$ on both sides, so we get that 
$R(\bigcap_{i=1}^{l} P_{i}^{-1}(s_i))\in R(\|_{i=1}^{l} \|_{j=1}^{k_i}L({\bf G}_{ij})) $. It follows by lemma~\ref{rp-relation} that
$$R(\bigcap_{i=1}^{l} P_{i}^{-1}(s_i))=\bigcap_{i=1}^{l} R(P_{i}^{-1}(s_i))=\bigcap_{i=1}^{l}P_{R,i}^{-1}R(s_i) \in R(\|_{i=1}^{l} \|_{j=1}^{k_i}L({\bf G}_{ij})) .$$ We know that  $\bigcap_{i=1}^{l}P_{R,i}^{-1}R(s_i)=\|_{i=1}^{l}R(s_i)=\|_{i=1}^{l}P_{R,i}(t)\supseteq t$ by $R(s_i)=P_{R,i}(t)$. We recall that $R(\|_{i=1}^{l} \|_{j=1}^{k_i}L({\bf G}_{ij}))\subseteq R(\|_{i=1}^{l} \|_{j=1}^{n_i}L({\bf G}_{ij}))=R(L({\bf G}))$ always holds under assumption (SEF). We thus get that $s\in R(L({\bf G}))$.  Therefore, $L({\bf M}) \subseteq  RL({\bf G})$ is proved under the assumption (SEF).
\end{proof}

\subsection{observability between the template level and the original system level}
For scalability reasons the controller synthesis  should be done only at the template level, in this subsection we will study the preservation of the property that is essential for computing supervisors under partial observation, i.e. observability between the template level and the original system level.
Before giving our result, we need the following Lemma known as transitivity of observability, which will be used to proof our main result (Theorem~\ref{main2}).
\begin{lemma} \label{transitivity}
Suppose $K\subseteq N\subseteq L=L({\bf G})\subseteq \Sigma^*$. If $K$ is observable with respect to $N$ and $N$ is observable with respect to $L$, then $K$ is observable with respect to $L$.
 \end{lemma}
\begin{proof} Let $s,\ s'\in \overline{K}$, $P(s)=P(s')$, $b\in \Sigma$, $sb\in \overline{K}$, and $s'b\in L$. We need to show that $s'b\in \overline{K}$. Since $K\subseteq N$, i.e. $\overline{K}\subseteq \overline{N}$ as well, we have $s,\ s'\in \overline{N}$. It now follows form the observability of $N$ with respect to $L$ that $s'b\in \overline{N}$. Finally, it is obtained from the observability of $K$ with respect to $N$ that $s'b\in \overline{K}$, which shows that $K$ is observable with respect to $L$. 
\end{proof}
Now  the result about preserving  observability is ready to be stated. We assume that the specification $E$ and the plant $L({\bf G})$ are non conflicting, i.e. $\overline{E\cap L({\bf G})}=\overline{E}\cap \overline{L({\bf G})}$.

  \begin{theorem}
  \label{main2}
    Let $L$ be a generator language over an event set $\Sigma$ with relabeling $R(L)$ over an event set $T$ and let $E\subseteq \Sigma^*$ be a  $({\bf G},R)$-normal specification and
$L({\bf M})$ is observable with respect to $R(L({\bf G}))$ and $P_R$.  If  $R(E)$ is observable with respect to $L({\bf M})$ and $P_R$, then $E\cap L({\bf G})$ is observable with respect to $L({\bf G})$ and $P$.
  \end{theorem}
  \begin{proof}
Let $R(E)$ be observable with respect to $L({\bf M})$ and $P_R$.
		It will be shown that $E\cap L({\bf G})$ is observable with respect to $L({\bf G})$ and $P$.
		Assume that $s,s'\in \overline{E\cap L({\bf G})}$, for some
		$b\in \Sigma$, such that  $sb\in \overline{E\cap L({\bf G})}$, $s'b\in L({\bf G})$, and $P(s)=P(s')$. We have to show that
		$s'b\in \overline{E\cap L({\bf G})}$.
 We have then $R(s),\ R(s') \in R(\overline{E\cap L({\bf G})})=\overline{R(E\cap L({\bf G}))}\subseteq \overline{R(E)}$, $R(P(s))=R(P(s'))$.  Let us denote $R(b)=a\in T$, then $R(s)a\in  \overline{R(E\cap L({\bf G}))}\subseteq \overline{R(E)}$ and $R(s')a\in  R(L({\bf G}))$.  Since $R$ preserves observability status of events in $\Sigma$, we have $P_R(R(s))=P_R(R(s'))$.  

From Lemma~\ref{transitivity} it follows that $R(E)$ is observable with respect to $R(L({\bf G}))$ and $P_R$.
Using this observability we conclude that $R(s')a\in  \overline{R(E)}$. It  implies that $s'b\in R^{-1}R(s')a\subseteq R^{-1}\overline{R(E)}=\overline{R^{-1}R(E)}$. We thus have $s'b\in \overline{R^{-1}R(E)}\cap L({\bf G})=\overline{R^{-1}R(E) \cap L({\bf G})}=\overline{E \cap L({\bf G})}$ by  $({\bf G},R)$-normality of $E$ assumption. Now the observability of $E$ with respect to $L({\bf G})$ is proved. 
  \end{proof} 
 \section{Efficient Verification of Sufficient Conditions in Theorem~\ref{Theorem:main1} and Theorem~\ref{maximal permissiveness}  }. In this section we use assumption(SEF), i.e. shared events between different agents in the same group are excluded.
We first address the verification of the sufficient condition
used in Theorem~\ref{Theorem:main1}. As shown in section~\ref{mainresult}, this condition requires computation of ${\bf G}$ which is dependent on the number of agents. Thus Proposition~\ref{prop:check} is proposed to avoid computing ${\bf G}$.  To prove Proposition~\ref{prop:check} we need the following lemmas and propositions. 


\begin{lemma} \label{Lemma:check0}
  For each group $i \in \{1,\ldots,l\}$ and arbitrary ${\bf G}_{i_1},{\bf G}_{i_2}, \in \mathcal{G}_i$,  if $L({\bf H}_{i})$ is relatively observable  with respect to $R(P_i(E)\cap L({\bf G}_{i_1}\|{\bf G}_{i_2}))$ and \\$R(L({\bf G}_{i_1}\|{\bf G}_{i_2}))$, then $L({\bf H}_{i})$ is relatively observable  with respect to $R(P_i(E)\cap L({\bf G}_{i_1}\|{\bf G}_{i_2}\|{\bf G}_{i_3}))$ and  $R(L({\bf G}_{i_1}\|{\bf G}_{i_2}))$.
 \end{lemma}
\begin{proof} It follows directly from definition of relative observability. 
\end{proof}

 In some proof below, we employ the same state transition structure of isomorphic generators to find a new string $t'$ that has a similar property with $t$, but with one less agent. For all agents in the same group we have the following "go down" property, an illustration is below the remark. 
\begin{remark}\label{remark}
Consider arbitrary $\ell$ agents ${\bf G}_{i_1},{\bf G}_{i_2},\dots,
{\bf G}_{i_\ell}\in \mathcal{G}_i$. If we have a string $t\in R(L(\|_{j=1}^{\ell}{\bf G}_{i_j}))$ and $t\notin R(L(\|_{j=1}^{\ell-1}{\bf G}_{i_j}))$ for $\ell\in \{ 1,\ldots, n_i\}$, then there exists a string $s\in \Sigma^*$ such that $s\in L(\|_{j=1}^{\ell}{\bf G}_{i_j})$ and for all $\tilde s$ with $R(\tilde s)=t$
including $s$ itself we have $\tilde s \notin L(\|_{j=1}^{\ell-1}{\bf G}_{i_j})$. Let $s'=\overline{P}_{i\ell}(s)$ for $\overline{P}_{i\ell}: \Sigma_i^*\to (\Sigma_i\setminus \Sigma_{i\ell})^*$. We thus have $s'\in L(\|_{j=1}^{j=\ell-1}{\bf G}_{i_j})$ and $s'\notin L(\|_{j=1}^{j=\ell-2}{\bf G}_{i_j})$. The latter is obtained from that if $s'\in L(\|_{j=1}^{\ell-2}{\bf G}_{i_j})$, then we have $s\in L(\|_{j=1}^{\ell-1}{\bf G}_{i_j})$ which is conflict with $t\notin R(L(\|_{j=1}^{\ell-1}{\bf G}_{i_j}))$. 
 Then we denote a string $t'\in T_i^*$ such that $R(s')=t'$. It implies that
 $t'\in R(L(\|_{j=1}^{\ell-1}{\bf G}_{i_j}))$ and $t'\notin R(L(\|_{j=1}^{\ell-2}{\bf G}_{i_j}))$. If there exist strings $s''$ with $R(s'')=t'$ such that $s''\in L(\|_{j=1}^{\ell-2}{\bf G}_{i_j})$, which implies that $t'\in R(L(\|_{j=1}^{j={\ell-2}}{\bf G}_{i_j}))$. Then from $s'\notin L(\|_{j=1}^{\ell-2}{\bf G}_{i_j})$ we see that $s''\neq s'=\overline{P}_{i\ell}(s)$. Therefore, for $s'=\overline{P}_{i\ell}(s)$ we have  $R(s')=t'\in R(L(\|_{j=1}^{\ell-1}{\bf G}_{i_j}))$ and $t'\notin R(L(\|_{j=1}^{\ell-2}{\bf G}_{i_j}))$.
\end{remark} 

For the small factory example in Fig. 3. If we take $t=11.10.11.10.11.10=R(111.110.121.120.131.130)$, then we have $t\in R(L({\bf G}_{i_1}\|{\bf G}_{i_2}\|{\bf G}_{i_3}))$, \\
$t \notin R(L({\bf G}_{i_1}\|{\bf G}_{i_2}))$, and string $s=111.110.121.120.131.130$ with $R(s)=t$.  Let $s'=\overline{P}_{i_3}(s)=111.110.121.120$. Then we have $t'=R(s')=R(111.110.121.120)=11.10.11.10$ for $\overline{P}_{i_3}: \Sigma_i^*\to (\Sigma_i\setminus \Sigma_{i_3})^*$. It can be verified that $t' \in R(L({\bf G}_{i_1}\|{\bf G}_{i_2}))$ and $t' \notin R(L({\bf G}_{i_1}))$.  
Now we proceed in a similar way with increasing the plant component.
\begin{lemma} \label{Lemma:check1}
  For each group $i \in \{1,\ldots,l\}$ and ${\bf G}_{i_1},{\bf G}_{i_2}, {\bf G}_{i_3}\in \mathcal{G}_i$,  if $L({\bf H}_{i})$ is relatively observable  with respect to $R(P_i(E)\cap L({\bf G}_i))$ and $R(L({\bf G}_{i_1}\|{\bf G}_{i_2}))$ for $P_i:\Sigma^*\rightarrow \Sigma_i^*$, then $R(L({\bf G}_{i_1}\|{\bf G}_{i_2}))$ is relatively observable  with respect to $R(P_i(E)\cap L({\bf G}_i))$ and  $R(L({\bf G}_{i_1}\|{\bf G}_{i_2}\|{\bf G}_{i_3}))$. 
 \end{lemma}
The following claim is needed.
 \begin{proposition} \label{prop0}
   For each group $i \in \{1,\ldots,l\}$ and ${\bf G}_{i_1},{\bf G}_{i_2}\in \mathcal{G}_i$,  if $L({\bf H}_{i})$ is relatively observable  with respect to $R(P_i(E)\cap L({\bf G}_i))$ and $R(L({\bf G}_{i_1}\|{\bf G}_{i_2}))$ for $P_i:\Sigma^*\rightarrow \Sigma_i^*$, then $L({\bf M}_{i})$  is relatively observable with respect to $R(P_i(E)\cap L({\bf G}_i))$ and $L({\bf G}_{i})$.
  \end{proposition}
 \begin{proof} 
Extending Lemma~\ref{Lemma:check1} inductively, it is derived that if $L({\bf H}_{i})$  ($i \in \{1,\ldots,l\}$) is  relatively observable with respect to $R(P_i(E)\cap L({\bf G}_i))$ and $R(L({\bf G}_{i_1}\|{\bf G}_{i_2}))$, then $L({\bf M}_{i})$ is  relatively observable with respect to $R(P_i(E)\cap L({\bf G}_i))$ and \\
$R(L(||_{j \in \{1,\ldots,k_{i}+1\}} {\bf G}_{i_j}))$. Again by applying the transitivity lemma (Lemma~\ref{transitivityRO}) inductively, we have $L({\bf M}_{i})$  is relatively observable with respect to $R(P_i(E)\cap L({\bf G}_i))$ and $L({\bf G}_{i})$. 
 \end{proof}
Now we are ready to prove Proposition~\ref{prop:check}.\\
   \begin{proof} 
Let $t,t'\in T^*$, $a\in T$, $ta\in  L({\bf M})$, $P_R(t)=P_R(t')$,  $t'\in R(P_i(E)\cap L({\bf G}))$, and $t'a \in R(L({\bf G}))$. We will show that $t'a\in  L({\bf M})$. 
 From $ta\in L({\bf M})$ we derive
\begin{align*}
&ta\in L({\bf M})= L(||_{i=1,\ldots, l} {\bf M}_{i})=\bigcap_{i=1,\ldots, l} P_{i|R}^{-1}L({\bf M}_{i}),
\end{align*}
where $P_{i|R}:T^*\rightarrow T_i^*.$ We thus get that $P_{i|R}(ta)\in L({\bf M}_i)$. 
Since $t'\in R(E\cap  L({\bf G}))$, there exists a string $s'\in \Sigma^*$ such that $s'\in E\cap  L({\bf G})$ and $R(s')=t'$. By $L({\bf G})=L(||_{i=1,\ldots, l} {\bf G}_{i})=\bigcap_{i=1,\ldots, l}P_i^{-1}(L({\bf G}_{i}))$ with $P_i:\Sigma^*\rightarrow \Sigma_i^*$, we have 
\begin{align*}
&s'\in E\cap \bigcap_{i=1,\ldots, l}P_i^{-1}(L({\bf G}_{i}))\subseteq  E\cap  P_i^{-1}(L({\bf G}_{i})).
\end{align*}
Then we get that $P_i(s')\in P_i(E)\cap  L({\bf G}_{i})$ which implies that $R(P_i(s'))=P_{i|R}(R(s'))=P_{i|R}(t')\in R(P_i(E)\cap  L({\bf G}_{i}))$. Recall that $t'a \in R(L({\bf G}))$, then there exists an event $b'\in \Sigma$ such that $R(b')=a$ and
\begin{align*}
&s'b'\in \|_{i=1}^l L({\bf G}_{i})=\bigcap_{i=1}^l P_{i}^{-1}(L({\bf G}_{i})). 
\end{align*}
We have $s'b'\in P_{i}^{-1}(L({\bf G}_{i}))$, i.e. $P_{i}(s'b')\in L({\bf G}_{i})$. Then  $R(P_{i}(s'b'))=P_{i|R}(R(s'b'))=P_{i|R}(t'a)\in R(L({\bf G}_{i}))$. It follows form $P_R(t)=P_R(t')$ that $P_R(P_{i|R}(t))=P_R(P_{i|R}(t'))$.
By Proposition~\ref{prop0} it directly follows that if $L({\bf H}_{i})$ ($i \in \{1,\ldots,l\}$) is relatively observable  with respect to $R(P_{i}(E)\cap L({\bf G}_{i_1}\|{\bf G}_{i_2}))$ and $R(L({\bf G}_{i_1}\|{\bf G}_{i_2}))$, then $L({\bf M}_i)$ is relatively observable  with respect to \\$R(P_{i}(E)\cap L({\bf G}_i))$ and $R(L({\bf G}_i))$. Therefore, $P_{i|R}(t'a)\in L({\bf M}_{i})$, i.e.  $t'a \in P_{i|R}^{-1}L(({\bf M}_{i}))$. It is derived that $t'a \in ||_{i=1}^l L({\bf M}_{i})=L({\bf M})$.
  \end{proof}



Next we will show that the  conditions in  Theorem~\ref{maximal permissiveness} can be checked with low computation effort. 

 Under shared event free (SEF) assumption between different agents in the same group relabeled observation consistency (ROC) condition in Theorem~\ref{maximal permissiveness}  always holds true. We need the following lemma that is used in the proof of Proposition \label{prop:checkingROC}.
 It is technical and states that it is possible to replace a string $s'$ by a similar string $w'$ with the same projection as a given string $s$. 
\begin{lemma} \label{lem:5}
Consider ${\bf G}$ given in (\ref{gi}). Consider  strings $s=\tilde s\sigma b\in L({\bf G})$ and $s'=\tilde s'\sigma' b'\in L({\bf G})$ with  $\sigma, \sigma'\in \Sigma_{uo}^*$ such that $P(\tilde s)=P(\tilde s')$ and $R(b)=R(b')$   Then there exists a string $w'\in L({\bf G})$ such that $P(w')=P(s)$ and $R(w')=R(s')$.
 \end{lemma}
\begin{proof} First of all, from $R(b)=R(b')$ it follows that $b,b'$ are either both observable
or both unobservable. Since $P(s)=P(\tilde s\sigma b)=P(\tilde s) P( b)$ and $P( s')=P(\tilde s'\sigma' b')=P(\tilde s') P(b')$, we have that $b,b'\in \Sigma_{uo}$ implies $P(s)=P(s')$. Then we can take $w'=s'$. Similarly, if $b,b'\in \Sigma_{o}$ and  $b=b'$ we also get
$P(s)=P(s')$.

Assume that $b\not=b'$ and $b,b'\in \Sigma_{o}$. Then $P(s)\not=P(s')$.\\
Let $\tilde s'=\alpha_1\beta_1\alpha_2\beta_2,\dots,\alpha_n\beta_n$ for $\alpha_1,\dots,\alpha_n\in \Sigma_{uo}$ and $\beta_1,\dots,\beta_n\in\Sigma_o$. Now we consider a  set
$$W=\{w\in R^{-1}R(\alpha_1)\beta_1R^{-1}R(\alpha_2)\beta_2,\dots,R^{-1}R(\alpha_n)\beta_n\}.$$
By similar structure of agents in groups we have from
$s'=\tilde s'\sigma' b'\in L({\bf G})$ that $WR^{-1}R(\sigma')b\cap L({\bf G})\neq \emptyset$. Thus, there  exists a string $w'\in WR^{-1}R(\sigma')b\cap L({\bf G})$, i.e.
$w'=\tilde w'\tilde\sigma'b \in L({\bf G})$ with $\tilde w'\in W$. We then have: 
$P(\tilde w')=P(\tilde s')=\beta_1\beta_2,\dots,\beta_n$ and 
$R(\tilde w')=R(\tilde s')=R(\alpha_1)R(\beta_1)R(\alpha_2)R(\beta_2),\dots,R(\alpha_n)R(\beta_n). $
Note that  $\tilde\sigma'\in \Sigma_{uo}^*$, because $\tilde\sigma'\in R^{-1}R(\sigma')$. 
Hence, we have $P(w')=P(s)$ and due to $R(b)=R(b')$ we also have
$R(w')=R(s')$. 
\end{proof}
Now we are ready to state the result concerning ROC.
\begin{proposition} \label{prop:checkingROC}
Consider ${\bf G}$ given in (\ref{gi}). Under (SEF) assumption  ROC condition always holds. 
\end{proposition} 
\begin{proof}
The proof goes by structural induction with respect to the observable string
$P_{R}(Rs)=P_{R}(t')\in T_o^*$, where $s\in L({\bf G})=L$ and $t'\in R(L)$. The base step is proven below.
If $P_{R}(Rs)=P_{R}(t')=\varepsilon$ then we need to show that
there exists a string $s'\in L$ such that $R(s')=t'$ and $P(s)=P(s')$.
Since $R(\Sigma_o)=T_o$ we have from $t'=R(s')$ that $s'\in \Sigma_{uo}^*$. Similarly,
from $P_{R}(Rs)=\varepsilon$, i.e. $Rs\in T_{uo}^*$ we have that $s\in \Sigma_{uo}^*$, i.e. we always have that $P(s)=P(s')$. 
Thus, it suffices to take any string $s'\in L$ such that $t'=R(s')$ that exists from $t'\in R(L)$.

The induction  hypothesis  consists in assuming that the ROC condition
holds for all $\tilde s\in L$ and $\tilde t'\in R(L)$ such that 
$P_{R}(R(\tilde s) )=P_{R}(t')=w\in T_{o}^*$. 
In the induction step we will show that ROC condition also holds for $s\in L$ and $t'\in R(L)$ with $P_{R}(Rs)=P_{R}(t')=wa\in T_{o}^*$, namely that there exists a string $s'\in L$ such that $R(s')=t'$ and $P(s)=P(s')$. 

Note that every $t'\in R(L)$ with $P_{R}(t')=wa$ is of the form $t'=\tilde t' \tau' a$
for some $\tilde t'\in R(L)$  with  $P_{R}(\tilde t')=w$ and $\tau'\in T_{uo}^*$.
Similarly, by denoting $t=Rs$, we have
$t\in R(L)$ with $P_{R}(t)=wa$, i.e. $t$ is of the form $t=\tilde t \tau a$
for some $\tilde t\in R(L)$  with  $P_{R}(\tilde t)=w$ and $\tau\in T_{uo}^*$.
Therefore, $s$ can be decomposed as 
$s=\tilde s\sigma b$, where $R(\tilde s)=\tilde t$, $R(\sigma)=\tau$, and $R(b)=a$.
We recall that $P_{R}(\tilde t')=w=P_{R}(R(\tilde s))$. Therefore,
from the induction hypothesis we know that there exists a string $\tilde s' \in L$ such that $R(\tilde s')=\tilde t'$, and $P(\tilde s)=P(\tilde s')$.

We will first show existence of a string $s'\in L$ such that $R(s')=t'$ and $P(s)=P(s')$.
We will search 
$s'$  of the form $s'\in \tilde s' R^{-1}(\tau')R^{-1}(a)$ so that  $R(s')=t'$ holds.
At the end of the proof we will replace string $s'$ by $w'$ based on Lemma \ref{lem:5}.

We will show that there exists an event $b'\in \Sigma_o$ and  a string $\sigma'\in \Sigma_{uo}^*$ such that for $s'=\tilde s'\sigma' b'$ we have  $R(s')=t'$ and $P(s)=P(s')$.
 

From $R(\tilde s)=\tilde t$ and $R(\tilde s')=\tilde t'$ with 
$P_{R}(\tilde t)=P_{R}(\tilde t')$ we obtain that $R(P_R(\tilde s))=P_{R}(R(\tilde s))=P_{R}(R(\tilde s'))=R(P_R(\tilde s'))$.

Let us consider a candidate string $s'=\tilde s'\sigma' b'\in L$ with  $R(s')=t'=\tilde t' \tau' a$.
We will show that if $P(s')\not=P(s)$ then there exists a $w'\in \Sigma^*$ with $R(w')=R(s')$ and $P(s)=P(w')$. 
Since $R(s)=t$ and $R(s')=t'$, we must have that $R(\sigma b)=\tau a$ and
$R(\sigma' b')=\tau' a$. Since $P(\tilde s)=P(\tilde s')$, we must have  $P(b')=P(\sigma' b')\not =P(\sigma b)=P(b).$ 
Therefore, $P(b)\not=P(b')$.
By Lemma~\ref{lem:5} from $R(b)=R(b')$ we obtain that by replacing $s'$ with $w'$, i.e. by taking
$w'=\tilde s'\tilde\sigma' b$ with $P(\tilde w')=P(\tilde s')$ and $R(\tilde w'\tilde\sigma' )=R(\tilde s'\sigma')=t'$ we have $w'\in L$ with $P(w')=P(s)$ and $R(w')=R(s')=t'$.
We then choose this $w'$ as the new $s'$ that satisfies the conditions.

\end{proof}

Finally, we will show how to efficiently check LROC condition in Theorem~\ref{maximal permissiveness}. We notice that due to
$R(b')=R(b)$ we necessarily have that $b$ and $b'$ belong to the same group of isomorphic agents, say $b,b'\in \Sigma_i=\bigcup_{j=1}^{j=n_i} \Sigma_{ij}$. 
If $b$ and $b'$ belong to the same agent then $b=b'$, in which case LROC is trivially satisfied.
Let us assume $b\in \Sigma_{ij}$ and $b'\in \Sigma_{ij'}$ for some $j\not=j'$.
Let $s,s'\in L$ with $P(s)=P(s')$ and  $b,b'\in \Sigma_{uo}$ such that $R(b)=R(b')$  then
$sb\in L$ means $P_{ij}(s)b\in L_{ij}$ and $s'b' \in L$ means  $P_{ij'}(s')b'\in L_{ij'}$
Verification of LROC condition then consists in checking $s'b \in L ,$ i.e. $P_{ij}(s')b\in L_{ij}$. This suggests that LROC condition is similar to observability. 
Indeed, observability of $P_{ij}^{-1}L_{ij}$
and $P_{ij'}^{-1}L_{ij'}$
means that for all $s,s'\in P_{ij}^{-1}(L_{ij})$ such that $P(s)=P(s')$,
$sa\in P_{ij}^{-1}(L_{ij})$ and $s'a\in P_{ij'}^{-1}(L_{ij'})$  observability means that
$s'a\in P_{ij}^{-1}(L_{ij})$.

Note that $s,s'\in L$ imply that in particular $s,s'\in P_{ij}^{-1}(L_{ij})$. $P_{ij}(s)b\in L_{ij}$ mean that $sb\in P_{ij}^{-1}(L_{ij})$, because $P_{ij}(sb)=P_{ij}(s)b\in L_{ij}$. Similarly, $s'b' \in L$, i.e. $P_{ij'}(s')b'\in L_{ij'}$ implies that $s'b'\in P_{ij'}^{-1}(L_{ij'})$, because $P_{ij'}(s'b')=P_{ij'}(s')b'\in L_{ij'}$.
Finally, checking LROC consists in checking if $P_{ij}(s')b\in L_{ij}$, which is
equivalent to $s'b\in P_{ij}^{-1}(L_{ij})$ using the same argument (namely $P_{ij}(b)=b\in \Sigma_{ij}$).

It follows from the above analysis that LROC can be checked 
in the same way as observability of $P_{ij}^{-1}L_{ij}$
with respect to $P_{ij'}^{-1}L_{ij'}$, where string $s$ is extended by $b$, but we allow instead of the same $b$ the event
$s'$ to be extended by a different event $b'$ with the same relabeling. 
Moreover, LROC can be viewed as relabeling counterpart of a similar condition
from hierarchical supervisory control with partial observations, called local observational
consistency (LOC), that was shown checkable in \cite{WODES20}.

\section{Conclusion}
We have studied multi-agent DES with partial observation, where the agents can be divided into several groups, and within each group the agents have similar state transition structures and can be relabeled into the same template. We have designed a scalable supervisor under partial observation whose state size and computational cost are independent of the number of agents. We have compared permissiveness of the scalable supervisor with the monolithic supervisor, and have proposed sufficient conditions, which guarantee that our scalable least restrictive supervisor is not more restrictive than the monolithic one. Moreover, we have proved that all sufficient conditions proposed in this paper can be verified with low computational effort.  
In a future work we will integrate  these partial observation results with already existing results on complete observations and on nonblockingness. Note that this paper is based on relabeling based abstraction of modular (multi-agent) DES and relabeling is a special case of mask type abstraction
as well as natural projection is another special case of abstraction. By integrating the results we can obtain results for hierarchical control under general, mask based, abstraction that can both rename and delete events.

\end{document}